\documentclass[prd,showpacs,showkeys]{revtex4}
\usepackage{bm}
\usepackage{amssymb,amsmath,amsthm,amssymb}
\begin{document}
\title{
Buchdahl-like transformations for perfect fluid spheres
}
\author{Petarpa Boonserm}
\email{Petarpa.Boonserm@mcs.vuw.ac.nz}
\homepage{http://www.mcs.vuw.ac.nz/research/people/Pertarpa/name}
\affiliation{School of Mathematics, Statistics, and Computer Science, 
Victoria University of Wellington, PO Box 600, Wellington, New Zealand\\}
\author{Matt Visser}
\email{matt.visser@vuw.ac.nz}
\homepage{http://www.mcs.vuw.ac.nz/~visser}
\affiliation{School of Mathematics, Statistics,  and Computer Science, 
Victoria University of Wellington, PO Box 600, Wellington, New Zealand\\}
\date{2 July 2007; 
\LaTeX-ed \today} 
\begin{abstract}

\noindent
In two previous articles [Phys.~Rev.~{\bf D71} (2005) 124307 (gr-qc/0503007), 
and gr-qc/0607001] we have discussed several ``algorithmic''
techniques that permit one (in a purely mechanical way) to generate
large classes of general relativistic static perfect fluid spheres.
Working in Schwarzschild curvature coordinates, we used these
algorithmic ideas to prove several ``solution-generating theorems''
of varying levels of complexity. In the present article we consider
the situation in other coordinate systems: In particular, in general diagonal coordinates we shall generalize our previous theorems, in
isotropic coordinates we shall encounter a variant of the so-called
``Buchdahl transformation'', while in other coordinate systems (such as Gaussian polar coordinates, Synge isothermal coordinates, and Buchdahl coordinates) we
shall find a number of more complex ``Buchdahl-like
transformations'' and ``solution-generating theorems'' that may be
used to investigate and classify the general relativistic static
perfect fluid sphere. Finally by returning to general diagonal coordinates and making a suitable \emph{ansatz} for the functional form of the metric components we place the Buchdahl transformation in its most general possible setting.  
\end{abstract}

\pacs{04.20.-q, 04.40.Dg, 95.30.Sf }
\keywords{Fluid spheres; arXiv:yymm.nnnn [gr-qc]}

\maketitle
\newcommand{\diff}[1]{\ensuremath{\mathrm{d}{#1}}}
\newcommand{\dx}[1]{\diff{#1}}
\newcommand{\dr}{\ensuremath{\frac{\mathrm{d}\phantom{r}}{\dx{r}}}}
\newcommand{\eqprime}[1]{\tag{\ref{#1}$^\prime$}}
\newcommand{\dOne}{\ensuremath{\delta_1}}
\newcommand{\dTwo}{\ensuremath{\delta_2}}
\newcommand{\dThr}{\ensuremath{\delta_3}}
\newcommand{\gth}[1]{\textsf{T}_\textsf{#1}}
\newtheorem{theorem}{Theorem}
\newtheorem{lemma}{Lemma}
\newtheorem{corollary}{Corollary}
\newtheorem{definition}{Definition}
\def\d{{\mathrm{d}}}
\def\implies{\Rightarrow}
\def\arctanh{{\mathrm{arctanh}}}
\def\SIM{\triangleq}
\def\txt{\textstyle}
\def\lint{\hbox{\Large $\displaystyle\int$}} 
\def\hint{\hbox{\Huge $\displaystyle\int$}} 

\section{Introduction}

Understanding perfect fluid spheres in general relativity is an 
important topic, recently deemed worthy of a brief chapter in an updated edition of the premier book summarizing and surveying ``exact solutions'' in general relativity~\cite{exact}. For considerably more details, the extensive review in~\cite{Delgaty} is an essential reference. Physically, perfect fluid spheres are a first approximation to
realistic models for a general relativistic star. The study of perfect fluid spheres is a long-standing topic with a venerable history~\cite{exact, Delgaty, Skea, Buchdahl, Bondi, Wyman, Hojman-et-al}, and continuing interest~\cite{Martin0, Martin1, Rahman, Lake}. Nevertheless the topic continues to provide surprising and
novel results~\cite{Petarpa1, Petarpa2, Petarpa:thesis, Petarpa:Greece, Petarpa:MG11}. In particular, as derived in references~\cite{Petarpa1}
and~\cite{Petarpa2}, and as further described in references~\cite{Petarpa:thesis, Petarpa:Greece, Petarpa:MG11},   we have developed several ``algorithmic'' techniques that
permit one to generate large classes of perfect fluid spheres from
first principles in a purely mechanical way. We now generalize these
algorithmic ideas, originally derived for Schwarzschild curvature
coordinates, (and to a lesser extent isotoropic coordinates), to a number of other coordinate systems. This
sometimes leads to much simpler results, and sometimes more general results. In this current article we shall aim to:
\begin{enumerate}

\item 
Derive a
number of new and significant ``transformation theorems'' and ``solution generating
theorems'' that allow us to map perfect fluid spheres into perfect
fluid spheres. 

\item
Investigate regularity conditions at the centre of the fluid sphere from the geometrical, physical, and purely mathematical points of view.

\item Develop a coherent picture of how these various theorems and various coordinate systems inter-relate to one another. 

\end{enumerate}

\section{Strategy}

We start the discussion by considering static spherically symmetric
distributions of matter, which implies (purely by symmetry) that in orthonormal components
the stress energy tensor takes the specific form
\begin{equation}
\label{T_a_b}
T_{\hat a\hat b} = \left[ 
\begin{array}{cccc}
\rho&0&0&0\\ 
0&p_r&0&0\\ 
0&0&p_t&0\\ 
0&0&0&p_t 
\end{array}\right].
\end{equation}
If the matter is a perfect fluid, then in addition we must have
\begin{equation}
p_{r} = p_{t}.
\end{equation}
Invoking the Einstein equations, this then implies a differential
relation on the spacetime geometry arising from equating the
appropriate orthonormal components of the Einstein tensor
\begin{equation}
\label{ortho}
G_{\hat\theta\hat\theta} = G_{\hat r\hat r} = G_{\hat\phi\hat\phi}.
\end{equation}
Equivalently, one could work with the appropriate orthonormal
components of the Ricci tensor
\begin{equation}
\label{ortho2}
R_{\hat\theta\hat\theta} = R_{\hat r\hat r} = R_{\hat\phi\hat\phi}.
\end{equation}
In terms of the metric components, this now leads to an ordinary
differential equation [ODE], which constrains the spacetime geometry
for \emph{any} general relativistic static perfect fluid sphere.

This equation constraining the spacetime geometry of static perfect fluid spheres is now analyzed in several different coordinate systems: general diagonal coordinates, Schwarzschild curvature coordinates, isotropic coordinates, and lesser-known coordinate systems such as Gaussian polar coordinates, Synge isothermal coordinates, and Buchdahl coordinates.
In the various coordinate systems we discuss below, we exhibit the
metric (in the form of the line element), the orthonormal components ($G_{\hat r\hat r}$,
$G_{\hat\theta\hat\theta}$, and and $G_{\hat t\hat t}$) of the Einstein tensor, and the ODE resulting from the pressure
isotropy condition $G_{\hat r\hat r} = G_{\hat\theta\hat\theta}$. We
then classify this ODE as to its type, and whenever possible, derive
suitable transformation theorems and solution generating theorems.

To place the use of ``unusual'' coordinate systems in perspective, one
might observe that Finch and Skea \cite{Skea} estimate that about 55\%
of all work on fluid spheres is carried out in Schwarzschild curvature
coordinates, that isotropic coordinates account for about 35\% of related 
research, and that the remaining 10\% is spread over multiple
specialized coordinate systems. We take the viewpoint that the
``unusual'' coordinate systems are useful \emph{only} insofar as they enable
us to obtain particularly simple analytic results, and the central
point of this article is to see just how much we can do in this
regard. We also consider what conditions are required on the spacetime metric and transformation theorems to guarantee regularity and physicality of both the geometry and the matter distribution near the center of the configuaration.

\def\txt{\textstyle}
\def\lint{\hbox{\Large $\displaystyle\int$}} 
\def\hint{\hbox{\Huge $\displaystyle\int$}}  

\section{General diagonal coordinates}
We begin by setting the notation. Choose coordinates to put the metric into the form:
\begin{equation} 
\label{general_diagonal_1}
\d s^2 = - \zeta(r)^2 \; \d t^2 + {\d r^2\over B(r)} + R(r)^2 \d \Omega^2.
\end{equation}
A brief calculation  yields
\begin{equation}
G_{\hat r\hat r} = - \frac{1-B (R')^2}{R^2} +{2B \zeta' R'\over R\zeta},
\end{equation}
\begin{equation}
G_{\hat\theta\hat\theta} = G_{\hat\phi\hat\phi} = \frac{1}{2}\; 
\frac{2 B \zeta' R' + 2 \zeta R'' B + \zeta B' R' + 2 R \zeta'' B + R \zeta' B'}{R \zeta},
\end{equation}
and
\begin{equation}
G_{\hat t\hat t} = {1-B (R')^2 \over R^2} +{2B R''-B'R'\over R}.
\end{equation}

\subsection{ODEs}

If this spacetime geometry is to be a perfect fluid sphere, then we
must have $G_{\hat r\hat r} = G_{\hat\theta\hat\theta}$. This isotropy
constraint supplies us with an ODE, which we can write in the form:
\begin{eqnarray}
\label{ODE_for_B_diagonal_1}
[R(R\zeta)']B' +2[R R'' \zeta+R^2\zeta''-RR'\zeta'-(R')^2\zeta] B 
+ 2 \zeta=0.
\end{eqnarray}
In this form the ODE is a
first-order linear non-homogeneous ODE in $B(r)$, and hence \emph{explicitly}
solvable. (Though it must be admitted that the explicit solution is sometimes rather messy.) 
The physical impact of this ODE (\ref{ODE_for_B_diagonal_1}) is that it reduces the freedom to choose the three \emph{a priori}
arbitrary functions in the general diagonal spacetime metric (\ref{general_diagonal_1})
to two arbitrary functions, (and we still have some remaining coordinate freedom in the $r$--$t$ plane, which then allows us to specialize the ODE even further). 

The \emph{same} ODE  (\ref{ODE_for_B_diagonal_1}) can also be rearranged as:
\begin{eqnarray}
\label{ODE_for_B_diagonal_2}
[2 R^2 B] \zeta'' + [R^2 B'-2BR R' ] \zeta' 
+ [2B(R R''-[R']^2) + R B' R' +2 ] \zeta = 0.
\end{eqnarray}
In this form the ODE is a
second-order linear and homogeneous ODE in $\zeta(r)$. While this ODE is not explicitly solvable in closed form, it is certainly true that a lot is known about its generic behaviour.

We can also view the two equivalent ODEs  (\ref{ODE_for_B_diagonal_1}) and  (\ref{ODE_for_B_diagonal_2})  as an ODE for $R(r)$:
\begin{eqnarray}
\label{ODE_for_B_diagonal_3}
[2B\zeta] R R'' + [B'\zeta-2B\zeta']R R' -[2B\zeta](R')^2 +[2B\zeta''+B'\zeta'] R^2 + 2\zeta= 0.
\end{eqnarray}
Viewed in this manner it is a second-order nonlinear ODE of no discernible special form --- and this approach does not seem to lead to any useful insights.

\subsection{Solution generating theorems}

Two rather general  solution generating theorems can be derived for the general diagonal line element.

\begin{theorem} [\bf General diagonal 1]
Suppose we adopt general diagonal coordinates and suppose that $\{ \zeta(r), B(r), R(r) \}$ represents a perfect fluid sphere.
Define
\begin{equation}
\label{delta_0}
\Delta(r) = \lambda \left({R(r) \zeta(r) \over (R(r) \zeta(r))'} \right)^2 \, \exp\left(2 \int{{\zeta'(r) \over \zeta(r)} \cdot {(R'(r) \zeta(r) - R(r) \zeta'(r)) \over (R'(r) \zeta(r) + R(r) \zeta'(r))}\, \d r} \right).
\end{equation}
Then for all $\lambda$, the geometry defined by holding $\zeta(r)$ fixed and
setting
\begin{equation} 
\label{general_diagonal_1b}
\d s^2 = - \zeta(r)^2 \; \d t^2 + {\d r^2\over B(r) + \Delta(r)} + R(r)^2 \d \Omega^2
\end{equation}
is also a perfect fluid sphere. That is, the mapping
\begin{equation}
\gth{gen1}(\lambda): \left\{ \zeta, B, R \right\} \mapsto 
\left\{ \zeta, B + \Delta, R \right\}
\end{equation}
takes perfect fluid spheres into perfect fluid spheres
\end{theorem}

\begin{proof}[Proof for Theorem 1]
The proof is based on the techniques used in \cite{Petarpa1, Petarpa:thesis}. No new principles are involved and we quickly sketch the argument.
Assume that $\left\{ \zeta(r), B(r), R(r) \right\}$ is a solution for equation (\ref{ODE_for_B_diagonal_1}).
Under what conditions does $\{ \zeta(r),\tilde{B}(r), R(r) \}$  also satisfy equation (\ref{ODE_for_B_diagonal_1})?
Without loss of generality, we write
\begin{equation}
\tilde{B}(r) = B(r) + \Delta(r).
\end{equation}
Substitute $\tilde{B}(r)$ in equation ({\ref{ODE_for_B_diagonal_1}})
\begin{eqnarray}
\label{ODE_for_B_diagonal_1_2}
&&\!\!\!\!\!\!\!\!
[R(R\zeta)'](B + \Delta)' +2[R R'' \zeta+R^2\zeta''-RR'\zeta'-(R')^2\zeta] (B + \Delta) 
 + 2 \zeta=0.
\end{eqnarray} 
Rearranging we see
\begin{eqnarray}
\label{ODE_for_B_diagonal_13}
&&\!\!\!\!\!\!\!\!
[R(R\zeta)']B' +2[R R'' \zeta+R^2\zeta''-RR'\zeta'-(R')^2\zeta] B 
 + 2 \zeta
\\
\nonumber
&&+[R(R\zeta)']\Delta' +2[R R'' \zeta+R^2\zeta''-RR'\zeta'-(R')^2\zeta]  \Delta=0.
\end{eqnarray}
But we know that the first line in equation (\ref{ODE_for_B_diagonal_13}) is zero,  because $\{R, \,\zeta,\,  B\}$ corresponds by hypothesis to a perfect fluid sphere.
Therefore
\begin{equation}
\label{ODE_for_B_diagonal_1_3}
[R(R\zeta)']\Delta' +2[R R'' \zeta+R^2\zeta''-RR'\zeta'-(R')^2\zeta]  \Delta=0,
\end{equation}
which is an ordinary homogeneous first order differential equation in $\Delta$.
A straightforward calculation, including an integration by parts, now leads to
\begin{equation}
\label{delta_01}
\Delta = \lambda \left({R \zeta \over (R \zeta)'} \right)^2 \, \exp\left(2 \int{{\zeta' \over \zeta} \cdot {(R' \zeta - R \zeta') \over (R' \zeta + R \zeta')}\, \d r} \right)
\end{equation}
as required.
\end{proof}

\begin{theorem}[\bf General diagonal 2]
Suppose we adopt general diagonal coordinates, and suppose that $\{ \zeta(r), B(r), R(r)\}$ represents a perfect fluid sphere. Define
\begin{equation}
Z(r) = \sigma + \epsilon \int{{R(r) \over \sqrt{B(r)} \, \zeta(r)^2} \, \d r}.
\end{equation}
Then for all $\sigma$ and $\epsilon$, the geometry defined by holding $B(r)$ and
$R(r)$ fixed and setting
\begin{equation} 
\d s^2 = - \zeta(r)^2 \, Z(r)^2\; \d t^2 + {\d r^2\over B(r)} + R(r)^2 \d \Omega^2
\end{equation}
is also a perfect fluid sphere.
That is, the mapping
\begin{equation}
\gth{gen2}(\sigma,\epsilon):
\left\{ \zeta, B, R \right\} \mapsto \left\{\zeta \; Z(\zeta, B, R), B, R \right\}
\end{equation}
takes perfect fluid spheres into perfect fluid spheres.
\end{theorem}
\begin{proof}
The proof is based on the technique of ``reduction in order'' as used in \cite{Petarpa1, Petarpa:thesis}. No new principles are involved and we quickly sketch the argument. Assuming that $\left\{ \zeta(r), B(r), R(r) \right\}$ solves equation (\ref{ODE_for_B_diagonal_2}), write
\begin{equation}
\zeta(r) \to \zeta(r) \; Z(r)\, .
\end{equation}
and demand that $\left\{ \zeta(r) \, Z(r), B(r), R(r) \right\}$ also
solves equation (\ref{ODE_for_B_diagonal_2}).  Then

\begin{equation}
\label{ODE_for_B_diagonal_2_22}
[2  R^2 \, B] \, (\zeta \, Z)''  + [R^2 \, B'-2 \, B \, R \, R' ] (\zeta \, Z)' + [2 B \, (R \, R''-[R']^2) + R \, B' R' +2 ] (\zeta \, Z) = 0.
\end{equation}
We expand the above equation to
\begin{eqnarray}
\label{ODE_for_B_diagonal_2_23}
\nonumber
&&[2 R^2 \, B] (\zeta'' \, Z + 2 \zeta' \, Z' + \zeta \, Z'')  + [R^2 \, B'-2 B \, R \, R' ] (\zeta' \, Z + \zeta \, Z' )
\\
&&
\qquad
+ [2 B \, (R \, R''-[R']^2) + R \, B' \, R' +2 ] (\zeta \, Z) = 0,
\end{eqnarray}
and then re-group to obtain
\begin{eqnarray}
\nonumber
&&\left\{[2 R^2 \, B] \, \zeta'' + [R^2 \, B'-2 \, B \, R \, R' ] \zeta' + [2 B \, (R \, R''-[R']^2) + R \, B' \, R' +2 ]\zeta \right\} Z
\\
&& 
\qquad
[2 R^2 \, B] (2 \zeta' \, Z' + \zeta \, Z'') + [R^2 \, B'-2 B \, R \, R']\zeta \, Z' = 0.
\end{eqnarray}
This is a linear homogeneous 2nd order ODE for $Z$. But under the current hypotheses the entire first line simplifies to zero  --- so the ODE now simplifies to
\begin{eqnarray}
&&[R^2 \, B' \, \zeta + 4 \, R^2 \, B \, \zeta' -2 B \, R \, R' \, \zeta] \, Z' + (2  R^2 \, B \, \zeta) \, Z'' = 0.
\end{eqnarray}
This is a first-order linear ODE in the dependent quantity $Z'(r)$.
Rearrange the above equation into
\begin{eqnarray}
Z'' \over Z' &=& - {[R^2 \, B' \, \zeta + 4 \, R^2 \, B \, \zeta' -2 B \, R \, R' \, \zeta] \over (2  R^2 \, B \, \zeta)}.
\end{eqnarray}
Simplifying
\begin{equation}
\label{ODE_2}
{Z'' \over Z'} = -{1 \over 2} \, {B' \over B} - 2 \, {\zeta' \over \zeta} + {R' \over R}.
\end{equation}
Re-write $Z'' / Z' = \d \ln (Z') / \d r$, and integrate twice over both sides of equation (\ref{ODE_2}), to obtain
\begin{equation}
Z(r) = \sigma + \epsilon \int{{R \over \sqrt{B} \, \zeta^2} \, \d r},
\end{equation}
depending on the old solution $\{\zeta, B, R\}$, and two arbitrary integration constants $\sigma$ and $\epsilon$.

\end{proof}

\subsection{Regularity conditions}
Now consider the result of placing various regularity conditions on the geometry of the perfect fluid sphere. (See particularly the discussion by Delgaty and Lake regarding the importance of regularity conditions~\cite{Delgaty}.) At  a minimum we want the spacetime geometry to be well-defined and smooth everywhere. This requires:
\begin{itemize}
\item 
$\zeta(r) > 0$, so that a clock held at any fixed but arbitrary ``place'' $(r, \theta, \phi)$ is well behaved, and at least continues to ``tick'' (even if it is fast or slow).

\item
Now  the ``centre'' of the spacetime is at some \emph{a priori} arbitrary coordinate $r_0$ such that the very specific condition $R(r_0) = 0$ is satisfied; by a simple coordinate change we can without loss of generality set $r_0 \rightarrow 0$ so that at the centre of the spacetime $R(0) = 0$. 

\item
So purely for geometrical reasons we should demand: 
\begin{itemize}
\item 
$R(0) = 0$;
\item
$\zeta(0) > 0$.
\end{itemize}

\item 
Now consider the spatial three-geometry
\begin{equation}
{\d r^2 \over B(r)} + R(r)^2 \, \d \Omega^2,
\end{equation}
and consider a small ``ball'' centered on the origin with coordinate radius $r$. The surface area of such a ball is
\begin{equation}
 S = 4 \, \pi \, R(r)^2,
\end{equation}
while its proper radius is
\begin{equation}
 \ell = \int_0^r {\d r \over \sqrt{B(r)}}.
\end{equation}
If the geometry at the centre is to be smooth (that is, locally Euclidean) then these two measures of radius must asymptotically agree (for small balls):
\begin{equation}
 \ell = R + \mathcal{O}(R^2).
\end{equation}
Thus in terms of the original coordinate $r$
\begin{equation}
\left.{\partial \ell\over\partial r}\right|_{r=0} = {1\over\sqrt{B(0)}} = \left.{\partial R\over\partial r}\right|_{r=0},
\end{equation}
implying
\begin{equation*}
B(0) \, [R'(0)]^2 = 1, 
\qquad \mathrm{and} \qquad B(0) \neq 0, 
\qquad \mathrm{that \, is} \qquad B(0) = 1/[R'(0)]^2.
\end{equation*}

\item Now if the spacetime is regular at the centre we must avoid ``kinks'' in the metric components.
Specifically, in terms of proper distance from the centre we should have  $\zeta(\ell)=\zeta(0)+  \mathcal{O}(\ell^2)$.
But in view of the already established relation between $\ell$ and $r$ this implies
$\zeta(r)=\zeta(0)+  \mathcal{O}(r^2)$ and we must have
\begin{equation}
\zeta'(0) = 0.
\end{equation}

\item
This does not complete the list of regularity conditions. We can also look at the orthonormal 
components of the Einstein tensor. Consider
\begin{equation}
G_{\hat r \hat r} = - \frac{1- B (R')^2}{R^2} + \frac{2 B \zeta' R' / \zeta}{R}.
\end{equation}
From the conditions already established, applying  the l'Hospital rule to the last term leads to a finite result $2B(0)\zeta''(0)/\zeta(0)$. For the remaining term,
noting that the denominator $(R^2)$ has a double zero at the centre, we deduce that the \emph{derivative} of the numerator must likewise be zero, in order that there be any chance for the l'Hospital rule to deliver a finite answer. 
This now yields the additional constraint that at the centre
\begin{equation}
[B (R')^2]' \to 0,
\end{equation}
that is
\begin{equation}
B'(0) R'(0) + 2 B(0) R''(0) = 0.
\end{equation}
In view of our earlier result for $B(0)$ this can be written as
\begin{equation}
B'(0) = - {2 R''(0) \over [R'(0)]^3} = \left.\left({1\over [R'(r)]^2}\right)'\right|_{r=0}.
\end{equation}

\item 
To check consistency, consider
\begin{equation}
G_{\hat r \hat r} + G_{\hat t \hat t} = \frac{2 B \zeta' R' / \zeta - B' R' - 2 B R''}{R},
\end{equation}
and note that this is finite at the centre provided $\zeta'\to0$ and $[B (R')^2]' \to 0$.

\item 
As a final consistency check consider
\begin{equation}
G_{\hat\theta\hat\theta} = G_{\hat\phi\hat\phi} = 
\frac{B\zeta'R'}{R\zeta} 
+
\frac{ 2  \zeta'' B +  \zeta' B'}{2 \zeta}
+
\frac{2 R'' B +  B' R' }{2R},
\end{equation}
and again note that this is finite at the centre provided $\zeta'\to0$ and $[B (R')^2]' \to 0$.

\item
So now let us  collect all our regularity conditions:
\begin{equation}
R(0) = 0, \qquad \zeta(0) > 0, \qquad B(0) = {1\over R'(0)^2} > 0, \qquad \zeta'(0) = 0 \qquad \mathrm{and} \qquad B'(0) = -{2 R''(0)\over[R'(0)]^3}.
\end{equation}

\item
Finally inserting all these conditions into the $r\to0$ limit we quickly find that at the centre of the fluid sphere the Einstein tensor has the limit:
\begin{eqnarray}
G_{\hat r \hat r}|_{\, 0} &=& G_{\hat \theta \hat \theta}|_{\, 0}= G_{\hat \phi \hat \phi}|_{\, 0}
= {B''(0) \over 2} + {2 \, \zeta''(0) \over \zeta(0) \, [R'(0)]^2} + {R'''(0) \over [R'(0)]^3} - {3 [R''(0)]^2 \over [R'(0)]^4},
\end{eqnarray}
\begin{equation}
G_{\hat t \hat t}|_{\, 0} = - {3 \over 2} \, \left\{ B''(0) + {2 \, R'''(0) \over [R'(0)]^3} - {6 [R''(0)]^2 \over [R'(0)]^4}\right\},
\end{equation}
and therefore
\begin{equation}
G_{\hat t \hat t}|_{\, 0} + 3 \, G_{\hat r \hat r}|_{\, 0} = \frac{6 \zeta''(0)}{\zeta(0) \, [R'(0)]^2}.
\end{equation}

\item
There is another largely independent way of getting these regularity conditions, either by hand or with the aid of a symbolic manipulation program such as {\sf Maple}. Let us start by assuming that near the origin all metric functions have Taylor series expansions
\begin{equation}
\zeta(r) = \zeta_0 + \zeta_1 \,r + {\zeta_2 \,r^2 \over 2!} + {\zeta_3 \,r^3 \over 3!}  + \mathcal{O}(r^4),
\end{equation}
\begin{equation}
B(r) = B_0 + B_1 \,r + {B_2 r^2\over2!} + {B_3 \,r^3\over3!} + \mathcal{O}(r^4),
\end{equation}
\begin{equation}
R(r) = R_1 \, r + {R_2 \,r^2\over2!} + {R_3 \,r^3 \over3!} + \mathcal{O}(r^4),
\end{equation}
Now use this \emph{ansatz} to calculate $G_{\hat t \hat t}$, $G_{\hat r \hat r}$, and $G_{\hat \theta \hat \theta}$ as Laurent series in $r$. A symbolic manipulation program such as {\sf Maple} will quickly tell you
\begin{equation}
G_{\hat a \hat b} = {A_{\hat a \hat b} \over r^2} + {B_{\hat a \hat b} \over r} + C_{\hat a \hat b} 
+ \mathcal{O}(r).
\end{equation}
Regularity now has the obvious meaning that you want the coefficients $A_{\hat a \hat b}$ and $B_{\hat a \hat b}$ of the $1/r^2$ and $1/r$ pieces to vanish  --- which implies constraints on the Taylor series coefficients $\{\zeta_i, \, B_i, \, R_i\}$. A brief computation should convince one that the pole residues $A_{\hat a \hat b}$ and $B_{\hat a \hat b}$ vanish provided:
\begin{equation}
R_0=0, \qquad \zeta_0>0, \qquad B_0=1/R_1^2 > 0, \qquad \zeta_1=0, \qquad \hbox{and} \qquad
B_1 = -2 R_2/(R_1)^3;
\end{equation}
that is
\begin{equation}
R(0)=0, \qquad \zeta(0) > 0, \qquad B(0) = 1/[R'(0)^2] > 0, \qquad \zeta'(0) = 0, \qquad \mathrm{and} 
\qquad B'(0) = -2 R''(0)/[R'(0)]^3.
\end{equation}
These are the same regularity conditions as we previously found by hand. Imposing these conditions, symbolic manipulation software such as {\sf Maple} can then easily evaluate the limits as you move to the centre of the star. 
\end{itemize}
The situation for general diagonal coordinates is  the most complicated case considered in this article, other cases discussed below are  simpler.

\subsection{Implications}

Note that the regularity conditions do feed back into theorem {\bf General diagonal 2}. In that theorem we were interested in the quantity
\begin{equation}
Z(r) = \sigma + \epsilon \int{{R(r) \over \sqrt{B(r)} \, \zeta(r)^2} \, \d r}.
\end{equation}
In view of the regularity conditions we now know that for a physically reasonable fluid sphere the integrand is finite all the way to the origin so that we can without loss of generality set
\begin{equation}
Z(r) = \sigma + \epsilon \int_0^r {{R(r) \over \sqrt{B(r)} \, \zeta(r)^2} \, \d r}.
\end{equation}
But to ensure regularity we want both $\zeta(0)$ and the new $\zeta(0)\to \zeta(0)\, Z(0)$ to be finite and positive, hence $Z(0)>0$, and therefore $\sigma>0$. This reproduces in a purely mathematical way one of the results of reference~\cite{Petarpa2}, where we used the TOV and boundedness of the central pressure and density to impose a similar constraint.

\subsection{Summary}

Let us now assess what we have done so far: For arbitrary spherically symmetric static spacetimes in general diagonal coordinates we have investigated the perfect fluid constraint, derived two ``solution generating theorems'', and applied regularity conditions to the spacetime. The analysis has so far been rather general, and we still have the freedom to make one further specific coordinate choice to simplify the line element, Einstein tensor, solution generating theorem, and regularity conditions.  Making the extra coordinate choice will sometimes simplify the mathematical structure sufficiently to allow more information to be extracted.

\section{Schwarzschild curvature coordinates} 
The Schwarzchild curvature coordinate system is by far the most popular coordinate system used in the study of perfect fluid spheres, and this coordinate choice (corresponding to $R(r)\to r$)  accounts for approximately 55\% of the research reported on perfect fluid spheres~\cite{Skea}.
The Schwarzschild coordinate system is the basis of the various theorems derived and discussed in
references~\cite{Petarpa1, Petarpa2,  Petarpa:thesis, Petarpa:Greece, Petarpa:MG11}.
The line element of a spherically symmetric spacetime in Schwarzschild curvature coordinates is
\begin{equation}
\d s^2 = - \zeta(r)^2 \; \d t^2 + {\d r^2\over B(r)} + r^2 \d\Omega^2.
\end{equation}
To begin, we calculate
\begin{equation}
G_{\hat r\hat r} = \frac{2 B \zeta' r - \zeta + \zeta B}{r^2 \zeta},
\end{equation}
\begin{equation}
G_{\hat\theta\hat\theta} = G_{\hat\phi\hat\phi} =\frac{1}{2}\; 
\frac{B' \zeta + 2 B \zeta' + 2 r \zeta'' B + r \zeta' B'}{r \zeta}.
\end{equation}
and
\begin{equation}
G_{\hat t \hat t} = -\frac{B' r - 1 + B}{r^2}.
\end{equation}

\subsection{ODEs}

Pressure isotrropy leads to the ODE
\begin{equation}
[r(r\zeta)']B'+[2r^2\zeta''-2(r\zeta)']B + 2\zeta=0.
\end{equation}
This is first-order linear non-homogeneous in $B(r)$.  Solving for
$B(r)$ in terms of $\zeta(r)$ is the basis of the analysis in
reference \cite{Martin1}, and is also the basis for Theorem {\bf 1} in
reference \cite{Petarpa1}. (See also~\cite{Petarpa:thesis} and the reports in~\cite{Petarpa:Greece, Petarpa:MG11}. After rephrasing in terms of the TOV
equation this is also related to Theorem {\bf P2} in reference
\cite{Petarpa2}). If we re-group in terms of $\zeta(r)$ we find
\begin{equation}
2 r^2 \zeta'' + (r^2 B'-2rB) \zeta' +(r B'-2B+2)\zeta=0,
\end{equation}
which is a linear homogeneous second-order ODE. This is the basis of
Theorem {\bf 2} in \cite{Petarpa1}, and after being recast as a Riccati
equation is also the basis of Theorem {\bf P1} in
reference~\cite{Petarpa2}.

\subsection{Regularity conditions}

Now consider the regularity conditions we demand to keep the spacetime regular at the centre.  First we have:
\begin{itemize}
\item
 The ``centre'' of the spacetime is by construction at  $r = 0$.
\item 
We must again demand $\zeta(r) > 0$, so that a clock held at fixed $(r, \theta, \phi)$ is well behaved and at least continues to ``tick'' (even if it is fast or slow).
\item
In particular $\zeta(0) > 0$.

\item
We now consider the spatial geometry
\begin{equation}
{\d r^2 \over B(r)} + r^2 \, \d \Omega^2,
\end{equation}
and slightly modify the discussion (previously developed in general diagonal coordinates) for a small ``ball'' centered on the origin with coordinate radius $r$. The surface area of such a ball is
\begin{equation}
 S = 4 \, \pi \, r^2,
\end{equation}
while its proper radius is
\begin{equation}
\ell = \int_0^r {\d r \over \sqrt{B(r)}}.
\end{equation}
If the geometry at the centre is to be smooth (locally Euclidean) then these two measures of radius must asymptotically agree (for small balls):
\begin{equation}
 \ell = r + \mathcal{O}(r^2).
\end{equation}
But in terms of the original coordinate $r$ this implies
\begin{equation}
B(0) \neq 0, \qquad \mathrm{and} \qquad \sqrt{B(0)} = 1, \qquad \mathrm{that \, \, is} \qquad B(0) = 1.
\end{equation}
Note we now have the much tighter constraint $B(0)=1$.

\item The same argument used for general diagonal coordinates still implies
\begin{equation}
\zeta'(0) = 0.
\end{equation}

\item 
By considering
\begin{equation}
G_{\hat r \hat r} = \frac{2 B \zeta'}{r \zeta} + \frac{B-1}{r^2}
\end{equation}
we deduce that at the centre
\begin{equation}
B' \to 0.
\end{equation}
That is
\begin{equation}
B'(0)  = 0.
\end{equation}

\item 
To check consistency note
\begin{equation}
G_{\hat r \hat r} + G_{\hat t \hat t} = \frac{-B' + 2 B \zeta' / \zeta}{r},
\end{equation}
which is finite at the centre provided $B' \to 0$ and $\zeta'\to 0$.

\item 
For the final consistency check note
\begin{equation}
G_{\hat\theta\hat\theta} = G_{\hat\phi\hat\phi} =
{B'\over2r} + {2B\zeta''+B'\zeta'\over2\zeta} + {B\zeta'\over r \zeta},
\end{equation}
which is finite at the centre provided $B' \to 0$ and $\zeta'\to 0$.

\item So now let's collect all our regularity conditions:
\begin{equation}
\zeta(0) > 0, \qquad B(0) = 1, \qquad \zeta'(0) = 0 \qquad \mathrm{and} \qquad B'(0) = 0.
\end{equation}

\item 
Inserting all these conditions into the limit we quickly find
\begin{equation}
G_{\hat r \hat r}|_{\, 0} = G_{\hat \theta \hat \theta}|_{\, 0} = G_{\hat \phi \hat \phi}|_{\, 0}
= {2 \, \zeta''(0) \over \zeta(0)} + {B''(0) \over 2},
\end{equation}
\begin{equation}
G_{\hat t \hat t}|_{\, 0} = -{3 B''(0) \over 2},
\end{equation}
and
\begin{equation}
G_{\hat t \hat t}|_{\, 0} + 3 \, G_{\hat r \hat r}|_{\, 0} = {6 \, \zeta''(0) \over \zeta(0)}.
\end{equation}

\item
The regularity regularity conditions derived above can again also be extracted in another manner --- by using a Taylor series expansion for the metric components to derive a Laurent series for the Einstein tensor --- using (for example) {\sf Maple},  this will still give the same regularity conditions as we derived by hand. This is a specialization of the method developed in the previous general diagonal coordinates. A little work with {\sf Maple} should convince you that the pole pieces vanish provided:
\begin{equation}
\zeta(0) > 0, \qquad B(0) = 1, \qquad \zeta'(0) = 0 \qquad \mathrm{and} \qquad B'(0) = 0.
\end{equation}

\end{itemize}
The ultimate consistency check on this formalism now results from considering general diagonal coordinates and simply setting $R(r)\to r$. That the results agree under this substitution gives us confidence that formalism is internally consistent.

\section{Isotropic coordinates} 

Isotropic coordinates are again commonly used coordinates when investigating perfect  fluid spheres. This is the second most popular coordinate system, accounting for about 35\% of published research in this field~\cite{Skea}. The spacetime metric for an arbitrary static spherically symmetric spacetime in isotropic coordinates is conveniently given by
\begin{equation} \label{Isotropic_1}
\d s^2 = - \zeta(r)^2 \; \d t^2 
+{1\over\zeta(r)^2\;B(r)^2} \{ \d r^2 + r^2 \d\Omega^2 \}.
\end{equation}
We calculate
\begin{equation}
G_{\hat r\hat r} = - 2 B B' \, \zeta^2/r + (B')^2 \zeta^2 - \zeta'^2 B^2,
\end{equation}
\begin{equation}
G_{\hat\theta\hat\theta} =  G_{\hat\phi\hat\phi} =
- B B' \zeta^2 /r + (B')^2 \, \zeta^2 -  B B'' \zeta^2 + B^2 \zeta'^2,
\end{equation}
and
\begin{equation}
G_{\hat t \hat t} =  2 B^2 \zeta \zeta'' + 4 B^2 \zeta \zeta'/r - 3 B^2 \zeta'^2 - 2 B B' \zeta \zeta' + 2 B B'' \zeta^2 - 3 B'^2 \zeta^2 + 4 B B' \zeta^2 /r .
\end{equation}

\subsection{ODEs}

The pressure isotropy condition leads to the very simple looking ODE~\cite{Rahman}:
\begin{equation} 
\label{ode_for_iso_zeta}
\left({\zeta'\over\zeta}\right)^2 ={B''-B'/r\over2B}.
\end{equation}
There are several ways of improving this. For instance, if we write
$\zeta(r)=\exp(\int g(r) \d r)$ then we have an algebraic equation for
$g(r)$~\cite{Rahman}:
\begin{equation}
g(r)= \pm \sqrt{B''-B'/r\over2B}.
\end{equation}
Conversely, the isotropy condition can be written in terms of $B(r)$ as:
\begin{equation}
\label{second_ode_iso}
B'' - B'/r -2g^2 B = 0.
\end{equation}
There is also an improvement obtained by writing $B(r)= \exp(2 \int
h(r) \d r)$ so that
\begin{equation}
g(r)^2 = 2 h(r)^2 + h'(r) - h(r)/r
\end{equation}
which is the basis of the analysis in reference~\cite{Rahman}.

\subsection{Solution generating theorems}

Let us now develop two transformation theorems appropriate to
isotropic coordinates.
\begin{theorem}[\bf Isotropic 1 --- Buchdahl transformation]
In Isotropic coordinates, if $\{\zeta(r),B(r)\}$ describes a perfect
fluid then so does $\{\zeta(r)^{-1},B(r)\}$. This is the well-known Buchdahl
transformation in disguise. That is, if
\begin{equation} 
\d s^2 = - \zeta(r)^2 \; \d t^2 
+{1\over\zeta(r)^2\;B(r)^2} \{ \d r^2 + r^2 \d\Omega^2 \}
\end{equation}
represents a perfect fluid sphere, then the geometry defined by
\begin{equation}
\d s^2 \to - \frac{1}{\zeta(r)^2} \; \d t^2 
+{\zeta(r)^2\over\;B(r)^2} \{ \d r^2 + r^2 \d\Omega^2 \}
\end{equation}
is also a perfect fluid sphere. Alternatively, the mapping
\begin{equation}
\gth{Iso 1}: \left\{ \zeta, B  \right\} \mapsto \left\{\zeta^{-1} , B \right\}
\end{equation}
takes perfect fluid spheres into perfect fluid spheres, and furthermore is a ``square root of
unity'' in the sense that:
\begin{equation}
\gth{Iso 1} \circ\gth{Iso 1} = I.
\end{equation}
\end{theorem}

\begin{proof}
By inspection. \emph{Vide} equation (\ref{ode_for_iso_zeta}).
\end{proof}

\begin{theorem}[\bf Isotropic 2]
Let $\{\zeta(r),B(r)\}$ describe a perfect fluid sphere. Define
\begin{equation}
Z =  \left\{\sigma +\epsilon \int {r \d r \over B(r)^2} \right\}.
\end{equation}
Then for all $\sigma$ and $\epsilon$, the geometry defined by holding
$\zeta(r)$ fixed and setting
\begin{equation} \label{Buchdahl}
\d s^2 = - \zeta(r)^2 \; \d t^2 +{1\over \zeta(r)^2 \;B(r)^2 \; Z(r)^2} \{ \d r^2 + r^2 \d\Omega^2 \}
\end{equation}
is also a perfect fluid sphere.
That is, the mapping
\begin{equation}
\gth{Iso 2}(\sigma,\epsilon): \left\{ \zeta, B  \right\} \mapsto \left\{ \zeta, B \; Z(B) \right\}
\end{equation}
takes perfect fluid spheres into perfect fluid spheres.
\end{theorem}

\begin{proof}
The proof is based on the technique of ``reduction in order'', and is a simple variant on the discussion in~\cite{Petarpa1}. Assuming that $\left\{ \zeta(r),B(r) \right\}$ solves equation (\ref{second_ode_iso}), write
\begin{equation}
B(r) \to B(r) \; Z(r) \, .
\end{equation}
and demand that $\left\{ \zeta(r),B(r) Z(r) \right\}$ also solves
equations (\ref{ode_for_iso_zeta}) and (\ref{second_ode_iso}).  We
find
\begin{equation}
(B \; Z)'' - (B \; Z)'/r -2g^2 (B\; Z) = 0.
\end{equation}
Re-grouping
\begin{equation}
\left\{B'' - B'/r -2g^2 B \right\} Z + (2 B' - B / r) Z' + B Z'' = 0.
\end{equation}
This is a linear homogeneous second-order ODE for $Z$ which now
simplifies [in view of (\ref{second_ode_iso})] to
\begin{equation}
\label{ode_iso_zeta}
(2 B' - B / r) Z' + B Z'' = 0 \, ,
\end{equation}
which is an ordinary homogeneous second-order differential equation,
depending only on $Z'$ and $Z''$. (So it can be viewed as a
first-order homogeneous order differential equation in $Z'$, which is
solvable.)  Separating the unknown variable to one side,
\begin{equation} 
\label{de_for_iso_zetaprime}
\frac{Z''}{Z'}=- 2 \frac{B'}{B} + \frac{1}{r} \, .
\end{equation}
Re-write $Z''/ Z' = \d\ln(Z')/\d r$, and integrate twice over both sides of
equation (\ref{de_for_iso_zetaprime}), to obtain
\begin{equation} \label{eq_for_iso_zeta_1}
Z(r) = \left\{\sigma +\epsilon \int {r \; \d r \over B(r)^2} \right\}.
\end{equation}
depending on the old solution $\left\{ \zeta(r), B(r) \right\}$, and
two arbitrary integration constants $\sigma$ and $\epsilon$.
\end{proof}

\begin{theorem}[\bf Isotropic 3]
  The transformations $\gth{Iso 1}$ and $\gth{Iso 2}$ commute.
\end{theorem}
\begin{proof}
  By inspection.
\end{proof} 

Note that the fact that these two transformation theorems commute is specific to isotropic coordinates. Such behaviour certainly does not occur in Schwarzschild coordinates or in general diagonal coordinates.

\subsection{Regularity conditions}

We now investigate the regularity conditions for isotropic coordinates. Some details are a little tricky, and have led to confusion in some of the previous published literature.
\begin{itemize}
\item
The ``centre'' of the spacetime is at some coordinate $r_0$; by a simple coordinate change we can set $r_0 \rightarrow 0$ so that at the centre of the spacetime is at $r=0$.

\item 
We still require $\zeta(r) > 0$, so that a clock held at fixed $(r, \theta, \phi)$ is well behaved and at least continues to ``tick'' (even if it is fast or slow). In particular
\begin{equation}
\zeta(0) > 0.
\end{equation}
\item
Now consider the spatial geometry
\begin{equation}
{\d r^2 \over \zeta(r)^2 \, B(r)^2} + {r^2 \over \zeta(r)^2 \, B(r)^2} \, \d \Omega^2 = 
{\left\{
{\d r^2} + {r^2} \, \d \Omega^2\right\}
\over  \zeta(r)^2 \, B(r)^2}
\end{equation}
and consider a small ``ball'' centered on the origin with coordinate radius $r$. The surface area of such a ball is
\begin{equation}
 S = 4 \, \pi \, {r^2 \over \zeta(r)^2 \, B(r)^2},
\end{equation}
while its proper radius is
\begin{equation}
 \ell = \int_0^r {\d r \over \zeta(r) \, B(r)}.
\end{equation}
If the geometry at the centre is to be smooth then these two measures of radius must asymptotically agree (for small balls):
\begin{equation}
\int_0^r {\d r \over \zeta(r) \, B(r)} = {r \over \zeta(r) \, B(r)} + \mathcal{O}(r^2),
\end{equation}
implying
\begin{equation}
{1 \over \zeta(r) \, B(r)} = \left({r \over \zeta(r) \, B(r)}\right)' + \mathcal{O}(r),
\end{equation}
that is
\begin{equation}
{1 \over \zeta(r) \, B(r)} = {1 \over \zeta(r) \, B(r)} + \mathcal{O}(r).
\end{equation}
This implies \emph{no constraint} on $\zeta(r)$ or $B(r)$, beyond the obvious constraint that these quantities remain finite and nonzero. Note in particular that in isotropic coordinates it is \emph{not true} that we need to enforce $B(0) = 1$. (If $B(0)\neq1$ we could think of arbitrarily rescaling it by a constant factor to force it to be unity, this would at worst multiply all curvature invariants by a constant, and so a nonsingular geometry remains nonsingular under such rescaling. This point is commonly misunderstood in the literature.)

\item
On the other hand, we could also find all our regularity conditions for isotropic coordinates by deducing them from the analysis developed for general diagonal coordinates. Identifying coefficients in the two metrics we see
\begin{equation}
\zeta_{\, iso} = \zeta_{\, gen}; \qquad\qquad (\zeta \, B)^2_{\, iso} = B_{\, gen},
\end{equation}
and
\begin{equation}
\label{R_gen}
{r \over (\zeta \, B)_{\, iso}} = R_{\, gen}.
\end{equation}
Now consider
\begin{equation}
[B (R')^2]_{gen} = (\zeta \, B)^2_{\, iso}  \left[ \left( {r \over (\zeta \, B)_{\, iso}} \right)'\right]^2 = 
\left[ 1 - {r\zeta'\over\zeta}-{r B'\over B} \right]_{iso}^2
\end{equation}
So at the centre $[B (R')^2]_{gen} \to 1$ provides no extra information about $B_{iso}$ or $\zeta_{iso}$ (beyond the fact that these quantities are finite and nonzero).

\item The same argument used for general diagonal coordinates and Schwarzschild coordinates still implies
\begin{equation}
\zeta'(0) = 0.
\end{equation}

\item
This does not complete the list of regularity conditions. Consider
\begin{equation}
G_{\hat r \hat r} =  -\frac{2B B' \zeta^2}{r}  + B'^2 \zeta^2  - \zeta'^2 B^2.
\end{equation}
Finiteness of this quantity now implies
\begin{equation}
B'(0) = 0
\end{equation}

\item 
On the other hand, finiteness of 
\begin{equation}
G_{\hat\theta \hat\theta} = G_{\hat\phi \hat\phi} = 
\frac{-B B' \zeta^2}{r} + B'^2 \zeta^2  - B B'' \zeta^2 + B^2 \zeta'^2,
\end{equation}
implies no additional constraint.

\item 
Similarly,  finiteness of 
\begin{equation}
G_{\hat t \hat t} =  2 B^2 \zeta \zeta'' + {4 B^2 \zeta \zeta'\over r} - 3 B^2 \zeta'^2 - 2 B B' \zeta \zeta' + 2 B B'' \zeta^2 - 3 B'^2 \zeta^2 + {4 B B' \zeta^2 \over r} ,
\end{equation}
implies no additional constraint.

\item
Summarizing, our regularity conditions are
\begin{equation}
B(0)>0; \qquad B'(0)=0; \qquad \zeta(0) >0; \qquad \zeta'(0)=0.
\end{equation}

\item
Now inserting the above conditions into the central limit we quickly find
\begin{eqnarray}
G_{\hat r \hat r}|_{\, 0} &=& G_{\hat \theta \hat \theta}|_{\, 0} =  G_{\hat \phi \hat \phi}|_{\, 0}
= 
- 2 B(0) B''(0) \zeta(0)^2,
\end{eqnarray}
\begin{equation}
G_{\hat t \hat t}|_{\, 0} = 6 B(0) B''(0) \zeta(0)^2 + 6 B(0)^2 \zeta(0) \zeta''(0),
\end{equation}
and
\begin{eqnarray}
\nonumber
G_{\hat t \hat t}|_{\, 0} + 3 \, G_{\hat r \hat r}|_{\, 0} 
&=& 6 B(0)^2 \zeta(0) \zeta''(0).
\end{eqnarray}



\item
On the other hand, the regularity conditions can also be easily obtained by using Taylor series for  the metric components to compute Laurent series for the orthonormal components of the Einstein tensor. A brief computation  with {\sf Maple} should convince you that the pole pieces vanish provided:
\begin{equation*}
\zeta(0) > 0, \qquad B(0) > 0, \qquad \zeta'(0) = 0 \qquad \mathrm{and} \qquad B'(0) = 0.
\end{equation*}
\end{itemize}

\subsection{Implications}

Note that the regularity conditions do feed back into theorem {\bf Isotropic 2}. In that theorem we were interested in the quantity
\begin{equation}
Z(r) = \sigma + \epsilon \int{{r \, \d r \over B(r)^2} }.
\end{equation}
In view of the regularity conditions we now know that for a physically reasonable fluid sphere the integrand is finite all the way to the origin so that we can without loss of generality set
\begin{equation}
Z(r) = \sigma + \epsilon \int_0^r {r \, \d r \over B(r)^2}.
\end{equation}
But to ensure regularity we want both $B(0)$ and the new $B(0)\to B(0)\, Z(0)$ to be finite and positive, hence $Z(0)>0$, and therefore $\sigma>0$.

\section{Gaussian polar coordinates}

Having now dealt with the most common coordinate systems, \emph{viz.}
Schwarzschild curvature coordinates and isotropic coordinates, we will
spend some time studying other more unusual coordinate systems.
Although Gaussian polar coodinates are typically less convenient than Schwarzschild or isotropic coordinates, we have however found them to be occasionally useful in the theory of symmetric static perfect fluid spheres.
Consider, for instance, the metric in Gaussian polar (proper radius)
coordinates:
\begin{equation}
\d s^2 = - \zeta(r)^2 \; \d t^2 +  \d r^2 + R(r)^2 \d\Omega^2.
\end{equation}
The following expressions can easily be calculated
\begin{equation}
G_{\hat r\hat r} = \frac{\zeta R'^2 - \zeta + 2 \zeta' R R'}{R^2 \zeta},
\end{equation}
\begin{equation}
G_{\hat\theta\hat\theta} = G_{\hat\phi\hat\phi} =\frac{\zeta R'' + \zeta' R' + \zeta'' R}{R \zeta},
\end{equation}
and
\begin{equation}
G_{\hat t \hat t} = \frac{ - R'^2 - 2 R R'' + 1}{R^2}.
\end{equation}
The isotropy condition supplies a second-order linear homogeneous ODE
for $\zeta(r)$:
\begin{equation} \label{second_ode_polar}
\zeta'' - \zeta' \; {R' \over R} + 
\zeta  \; \left\{ {1-(R')^2+ R''R\over R^2 }\right\} = 0.
\end{equation}
In contrast, when rearranged and viewed as an ODE for $R(r)$, 
\begin{equation}
\zeta R R'' - \zeta' R R' - \zeta (R')^2 + \zeta'' R^2 + \zeta=0,
\end{equation}
this has no particularly desirable
properties.

\subsection{Solution generating theorem}

\begin{theorem}[\bf Gaussian polar]
  Suppose we are in Gaussian polar coordinates, and that $\{
  \zeta(r), R(r) \}$ represents a perfect fluid sphere.  Define
\begin{equation}
\Lambda(r) =  \left\{\sigma +\epsilon \int {R(r) \d r \over \zeta(r)^2} \right\}.
\end{equation}
Then for all $\sigma$ and $\epsilon$, the geometry defined by holding
$R(r)$ fixed and setting
\begin{equation} 
\d s^2 = - \zeta(r)^2 \Lambda(r)^2\; \d t^2 +  \d r^2 + R(r)^2 \d\Omega^2 
\end{equation}
is also a perfect fluid sphere.
That is, the mapping
\begin{equation}
\gth{Gaussian}(\sigma,\epsilon):
\left\{ \zeta, R \right\} \mapsto \left\{\zeta \; \Lambda(\zeta, R) ,R \right\}
\end{equation}
takes perfect fluid spheres into perfect fluid spheres.
\end{theorem}

\begin{proof}
  The proof is again based on the technique of ``reduction in order''.
  No new principles are involved and the argument is quickly sketched.
  Assuming that $\left\{ \zeta(r), R(r) \right\}$ solves equation
  (\ref{second_ode_polar}), write
\begin{equation}
\zeta(r) \to \zeta(r) \; \Lambda(r),
\end{equation}
and demand that $\left\{ \zeta(r) \, \Lambda(r), R(r) \right\}$ also
solves equation (\ref{second_ode_polar}).  Then
\begin{equation}
(\zeta \; \Lambda)'' - (\zeta \; \Lambda)' {R' \over R} 
+ (\zeta \; \Lambda) \left\{ {1-(R')^2+ R''R\over R^2 }\right\} = 0,
\end{equation}
which  simplifies to
\begin{equation}
\label{ode_polar_zeta}
\left(2 \zeta' -  \zeta \; {R' \over R}\right) \Lambda' + \zeta \; \Lambda'' = 0.
\end{equation}
This is a 1st-order homogeneous order differential equation in
$\Lambda'$, which is solvable.  Separating the unknown variable, and
integrating twice one obtains
\begin{equation} \label{eq_for_polar_zeta_1}
\Lambda(r) = \left\{\sigma +\epsilon \int {R(r) \; \d r \over \zeta(r)^2} \right\}.
\end{equation}
depending on the old solution $\left\{ \zeta(r) , R(r) \right\}$, and
two arbitrary integration constants $\sigma$ and $\epsilon$.
\end{proof}

\subsection{Regularity conditions}

In the following section, we briefly describe the regularity conditions that must hold at the centre of the fluid sphere
\begin{itemize}
\item 
Adapting the arguments that by now should be completely standard we have
\begin{equation}
R(0)=0; \qquad \zeta(0) > 0; \qquad \zeta'(0)=0.
\end{equation}

\item
Now consider the spatial geometry
\begin{equation}
\d r^2 + R(r)^2 \, \d \Omega^2,
\end{equation}
and consider a small ``ball'' centered on the origin with coordinate radius $r$. The surface area of such a ball is
\begin{equation}
 S = 4 \, \pi \, R(r)^2,
\end{equation}
while its proper radius is
\begin{equation}
 \ell = \int_0^r 1 \, \d r = r .
\end{equation}
If the geometry at the centre is to be smooth, these two measures of radius must asymptotically agree (for small balls):
\begin{equation}
 r = R + \mathcal{O}(R^2).
\end{equation}
But in terms of the original coordinate $r$ this implies
\begin{equation}
R'(0) = 1.
\end{equation}

\item
Now consider 
\begin{equation}
G_{\hat r \hat r} = \frac{R'^2 -1}{R^2} + \frac{2 \zeta' R' / \zeta}{R}
\end{equation}
The last term is finite by the l'Hospital rule and the previously derived condition $\zeta'\to 0$. Applying the l'Hospital rule twice the first term will be finite at the centre iff
\begin{equation}
[(R')^2]' \to 0,
\end{equation}
that is, iff
\begin{equation}
R''(0)  = 0.
\end{equation}

\item
For a consistency check, note that 
\begin{equation}
G_{\hat r \hat r} + G_{\hat t \hat t} = \frac{2 \zeta' R' / \zeta - 2 R''}{R},
\end{equation}
which implies no new constraint.

\item
For a consistency check, note that 
\begin{equation}
G_{\hat \theta \hat \theta}  = G_{\hat \phi \hat \phi}  = 
{R''\over R} +{R'\zeta'\over R\zeta} + {\zeta''\over \zeta},
\end{equation}
which implies no new constraint.

\item
Now let us collect all our regularity conditions:
\begin{equation}
R(0) = 0,  \qquad R'(0) = 1, \qquad R''(0) = 0, \qquad 
\zeta(0) > 0, \qquad \mathrm{and} \qquad  \zeta'(0) = 0 .
\end{equation}

\item
Inserting all these conditions into the limit we quickly deduce
\begin{equation}
G_{\hat r \hat r}|_{\, 0} = G_{\hat \theta \hat \theta}|_{\, 0} = G_{\hat\phi \hat \phi}|_{\, 0}
=
R'''(0) + {2 \zeta''(0) \over \zeta(0)},
\end{equation}
\begin{equation}
G_{\hat t \hat t}|_{\, 0} = - 3 R'''(0),
\end{equation}
and
\begin{equation}
G_{\hat t \hat t}|_{\, 0} + 3 \, G_{\hat r \hat r}|_{\, 0} = {6 \, \zeta''(0) \over \zeta(0)}.
\end{equation}

\item
These regularity conditions can again be verified by using Taylor series in the metric components to calculate Laurent series for the orthonormal components of the Einstein tensor. Using {\sf Maple}, the pole pieces vanish provided:
\begin{equation}
\zeta(0) > 0, \qquad  \zeta'(0) = 0, \qquad R(0) = 0, \qquad R'(0) = 1,  
\qquad \mathrm{and} \qquad R''(0) = 0.
\end{equation}

\end{itemize}

\subsection{Implications}

Note that the regularity conditions do feed back into theorem {\bf Gaussian polar}. In that theorem we were interested in the quantity
\begin{equation}
Z(r) = \sigma + \epsilon \int{{R(r) \, \d r \over \zeta(r)^2} }.
\end{equation}
In view of the regularity conditions we now know that for a physically reasonable fluid sphere the integrand is finite all the way to the origin so that we can without loss of generality set
\begin{equation}
Z(r) = \sigma + \epsilon \int_0^r {R(r) \, \d r \over \zeta(r)^2}.
\end{equation}
But to ensure regularity we want both $\zeta(0)$ and the new $\zeta(0)\to \zeta(0)\, Z(0)$ to be finite and positive, hence $Z(0)>0$, and therefore $\sigma>0$.

\section{Synge isothermal coordinates}
Consider the metric in Synge isothermal (tortoise) coordinates:
\begin{equation}
\d s^2 = - \zeta(r)^{-2} \; \{ \d t^2 - \d r^2\} + \{\zeta(r)^{-2} \;R(r)^2 \d\Omega^2 \}.
\end{equation}
A brief computation yields
\begin{equation}
G_{\hat r\hat r} =  \frac{3  R^2 \; (\zeta')^2 - 4 R \zeta  R' \; \zeta' + (R')^2 \zeta^2 - \zeta^2}{R^2},
\end{equation}
\begin{equation}
G_{\hat\theta\hat\theta} =   \frac{3 R (\zeta')^2 - 2 \zeta R' \zeta' - 2 R \; \zeta \zeta'' + \zeta^2 R''}{R},
\end{equation}
and
\begin{equation}
G_{\hat t \hat t} =  \frac{-3 R^2 \zeta'^2 + 2 R^2 \zeta \zeta'' + 4 R R' \zeta \zeta' + \zeta^2 - R'^2 \zeta^2 - 2 R R'' \zeta^2}{R^2}.
\end{equation}
Pressure isotropy supplies us with an second-order homogeneous ODE for
$\zeta(r)$:
\begin{equation} \label{second_order_synge}
\zeta''-\zeta' {R'\over R} - \zeta \left\{ {1-(R')^2+R R''\over 2 R^2}  \right\} = 0.
\end{equation}
This is very similar to the isotropy equation in Gaussian polar coordinates.
On the other hand, rearranging to obtain an ODE for $R(r)$ yields a second order nonlinear differntial equation of no particularly useful form.

\subsection{Solution generating theorem}

Our first key result can be phrased as a simple theorem:
\begin{theorem}[\bf Synge]
Suppose $\{ \zeta(r), R(r) \}$ represents a perfect fluid sphere.
Define
\begin{equation}
A(r) =  \left\{\sigma +\epsilon \int {R(r) \d r \over \zeta(r)^2} \right\}.
\end{equation}
Then for all $\sigma$ and $\epsilon$, the geometry defined by holding
$R(r)$ fixed and setting
\begin{equation} 
\d s^2 = - \frac{1}{\zeta(r)^2 \; A(r)^2}\; \left\{\d t^2 - \d r^2\right\} 
+ \frac{R(r)^2}{\zeta(r)^2 \; A(r)^2} \d\Omega^2 
\end{equation}
is also a perfect fluid sphere. 
That is, the mapping
\begin{equation}
\gth{Synge}(\sigma,\epsilon): 
\left\{ \zeta, R \right\} \mapsto \left\{\zeta \; A(\zeta, R) ,R \right\}
\end{equation}
takes perfect fluid spheres into perfect fluid spheres.
\end{theorem}

\begin{proof}
  The proof is again based on the technique of ``reduction in order''.
  No new principles are involved and the argument is very quickly
  sketched.  Assuming that $\left\{ \zeta(r), R(r) \right\}$ solves
  equation (\ref{second_order_synge}), write
\begin{equation}
\zeta(r) \to \zeta(r) \; A(r) \, .
\end{equation}
and demand that $\left\{ \zeta(r)\, A(r), R(r) \right\}$ also solves
equation (\ref{second_order_synge}).  We eventually find
\begin{equation}
\label{ode_polar_zeta2}
\left(2 \zeta' -  \zeta \; {R' \over R}\right) A' + \zeta \; A'' = 0 \, ,
\end{equation}
which is solvable. Separating the unknown variable to one side, and
integrating twice we obtain
\begin{equation} 
\label{eq_for_synge_zeta_1}
A(r) = \left\{\sigma +\epsilon \int {R(r) \; \d r \over \zeta(r)^2} \right\}.
\end{equation}
depending on the old solution $\left\{ \zeta(r) , R(r) \right\}$, and
two arbitrary integration constants $\sigma$ and $\epsilon$.
\end{proof}

\subsection{Regularity conditions}

\begin{itemize}
\item
As usual
\begin{equation}
R(0)=0; \qquad \zeta(0)>0; \qquad \zeta'(0)=0.
\end{equation}

\item 
Now consider the spatial geometry
\begin{equation}
{\d r^2 \over \zeta(r)^2} + {R(r)^2 \over \zeta(r)^2} \, \d \Omega^2 = {\d r^2  + R(r)^2\, \d \Omega^2  \over \zeta(r)^2} 
\end{equation}
and consider a small ``ball'' centered on the origin with coordinate radius $r$. The surface area of such a ball is
\begin{equation}
 S = 4 \, \pi \, {R(r)^2 \over \zeta(r)^2},
\end{equation}
while its proper radius is
\begin{equation}
 \ell = \int_0^r {1 \over \zeta(r)} \, \d r .
\end{equation}
If the geometry at the centre is to be smooth then these two measures of radius must asymptotically agree (for small balls):
\begin{equation}
 \ell = {R \over \zeta} + \mathcal{O}[(R/\zeta)^2].
\end{equation}
But in terms of the original coordinate $r$ this implies
\begin{equation}
R'(0) = 1.
\end{equation}
Note that there is no additional constraint on $\zeta(0)$ apart from the fact that it is nonzero and finite.

\item
Consider
\begin{equation}
G_{\hat r \hat r} = 3 R'^2 - \frac{4 \zeta \zeta' R' }{R} + \frac{R'^2 \zeta^2 - \zeta^2}{R^2}
\end{equation}
Regularity at the centre requires
\begin{equation}
[(R')^2]' \to 0,
\end{equation}
which implies
\begin{equation}
R''(0)=0.
\end{equation}

\item 
The combination
\begin{equation}
G_{\hat r \hat r} + G_{\hat t \hat t} = 2\zeta \zeta'' - \frac{2R'' \zeta^2}{R},
\end{equation}
does not lead to any new constraint.

\item
The component
\begin{equation}
G_{\hat \theta \hat \theta}  = G_{\hat \phi \hat \phi}  = 
3(\zeta')^2 - 2\zeta \zeta'' + \frac{2 R'' \zeta^2-2\zeta R' \zeta' }{R}
\end{equation}
does not lead to any new constraint.

\item
Collecting all our regularity conditions:
\begin{equation}
\zeta(0) > 0, \qquad \zeta'(0) = 0, \qquad R(0) = 0, \qquad R'(0) = 1, 
\qquad  \mathrm{and} \qquad R''(0) = 0.
\end{equation}

\item
Finally inserting all these conditions into the limit we quickly find
\begin{equation}
G_{\hat r \hat r}|_{\, 0} = G_{\hat \theta \hat \theta}|_{\, 0} = G_{\hat \phi \hat \phi}|_{\, 0}
=- 4 \, \zeta(0) \zeta''(0) + \zeta(0)^2 R'''(0),
\end{equation}
\begin{equation}
G_{\hat t \hat t}|_{\, 0} = 6 \, \zeta(0) \zeta''(0) - 3 \, \zeta(0)^2 R'''(0),
\end{equation}
and
\begin{equation}
G_{\hat t \hat t}|_{\, 0} + 3 \, G_{\hat r \hat r}|_{\, 0} = -6 \, \zeta(0) \zeta''(0).
\end{equation}

\item
Again working with Taylor series to derive Laurent series,  the pole pieces vanish and the centre of the sphere is regular provided:
\begin{equation}
\zeta(0) > 0, \qquad  \zeta'(0) = 0, \qquad 
R(0) = 0, \qquad  R'(0) = 1, \qquad \mathrm{and} \qquad R''(0) = 0.
\end{equation}

\end{itemize}

\subsection{Implications}

Note that the regularity conditions do feed back into theorem {\bf Synge}. In that theorem we were interested in the quantity
\begin{equation}
A(r) = \sigma + \epsilon \int{{R(r) \, \d r \over \zeta(r)^2} }.
\end{equation}
In view of the regularity conditions we now know that for a physically reasonable fluid sphere the integrand is finite all the way to the origin so that we can without loss of generality set
\begin{equation}
A(r) = \sigma + \epsilon \int_0^r {R(r) \, \d r \over \zeta(r)^2}.
\end{equation}
But to ensure regularity we want both $\zeta(0)$ and the new $\zeta(0)\to \zeta(0)\, Z(0)$ to be finite and positive, hence $Z(0)>0$, and therefore $\sigma>0$.


\section{Buchdahl coordinates}
Without loss of generality we can put the metric in ``Buchdahl coordinates'' and choose the coefficients to be
\begin{equation}
\d s^2 = - \zeta(r)^{2} \; \d t^2 + \zeta(r)^{-2} \left\{ {\d r^2 } +  R(r)^2 \d\Omega^2 \right\}.
\end{equation}
This coordinate system is a sort of cross between Synge isothermal
(tortoise) coordinates and Gaussian polar (proper radius) coordinates.
We calculate
\begin{equation}
G_{\hat r\hat r} =  -(\zeta')^2 - {\zeta^2\;[1-(R')^2]\over R^2},
\end{equation}
\begin{equation}
G_{\hat\theta\hat\theta} =  G_{\hat\phi\hat\phi} = +(\zeta')^2 + {\zeta^2\;R''\over R} ,
\end{equation}
and
\begin{equation}
G_{\hat t \hat t} =  -3(\zeta')^2+2\zeta\zeta''- {2\zeta( \zeta R''-   2 R' \zeta' )\over R} 
+ {\zeta^2\;[1-(R')^2]\over R^2} .
\end{equation}
Imposing pressure isotropy supplies us with a particularly simple first-order homogeneous
ODE for $\zeta(r)$:
\begin{equation} 
\label{ode_for_Buchdahl}
\left({\zeta'\over\zeta}\right)^2 =-{[1-(R')^2+R R'']\over 2 R^2}.
\end{equation}
This is very similar to the equation we obtained in isotropic coordinates. (Rearranging this into an ODE for $R(r)$ yields a second-order nonlinear differential equation which does not seem to be particularly useful.)

\begin{theorem}[\bf Buchdahl]
  If $\{\zeta(r),R(r)\}$ describes a perfect fluid then so does
  $\{\zeta(r)^{-1},R(r)\}$. This is the Buchdahl transformation in yet
  another disguise.  The geometry defined by holding $R(r)$ fixed and
  setting
\begin{equation}
\d s^2 = - \zeta(r)^{-2} \; \d t^2 + {\zeta(r)^2}  \left\{ \d r^2 
+ {R(r)^2} \d\Omega^2 \right\}
\end{equation}
is also a perfect fluid sphere. 
That is, the mapping
\begin{equation}
\gth{Buchdahl}: \left\{ \zeta, R  \right\} \mapsto \left\{\zeta^{-1} , R \right\}
\end{equation}
takes perfect fluid spheres into perfect fluid spheres.
\end{theorem}

\begin{proof}
By inspection. (Note strong similarities to the discussion for isotropic
coordinates.)
\end{proof}

\subsection{Regularity conditions}

\begin{itemize}

\item
As usual
\begin{equation}
R(0)=0; \qquad \zeta(0)>0; \qquad \zeta'(0)=0.
\end{equation}

\item
Now consider the spatial geometry
\begin{equation}
\zeta(r)^{-2}\left\{ \, \d r^2 +  R(r)^2 \, \d \Omega^2\right\}
\end{equation}
and consider a small ``ball'' centered on the origin with coordinate radius $r$. 
The surface area of such a ball is
\begin{equation}
 S = 4 \, \pi \, \zeta(r)^{-2} \, R(r)^2,
\end{equation}
while its proper radius is
\begin{equation}
 \ell = \int_0^r \zeta(r)^{-1} \, \d r .
\end{equation}
If the geometry at the centre is to be smooth then these two measures of radius must asymptotically agree (for small balls):
\begin{equation}
 \ell = R/\zeta + \mathcal{O}[(R/\zeta)^2].
\end{equation}
But in terms of the original coordinate $r$ this implies
\begin{equation}
R'(0) = 1
\end{equation}

\item
Consider
\begin{equation}
G_{\hat r \hat r} =- (\zeta')^2 - \frac{\zeta^2[1-(R')^2]}{R^2} 
\end{equation}
this now yields the additional constraint
\begin{equation}
[(R')^2]' \to 0,
\end{equation}
which implies
\begin{equation}
R''(0)=0.
\end{equation}

\item
The combination
\begin{equation}
G_{\hat r \hat r} + G_{\hat t \hat t} = -{4(\zeta')^2} + {2 \zeta \zeta''} 
-{2\zeta( \zeta R''-2 \zeta' R') \over R},
\end{equation}
does not lead to any new constraint.

\item
Similarly
\begin{equation}
G_{\hat\theta\hat\theta} =   (\zeta')^2 + \frac{ \zeta^2 R''}{R},
\end{equation}
does not lead to any new constraint.

\item

So now let's collect all our regularity conditions:
\begin{equation}
R(0) = 0, \qquad R'(0) = 1, \qquad  R''(0) = 0, \qquad
\zeta(0) > 0, \qquad \mathrm{and} \qquad\qquad  \zeta'(0) = 0.
 \end{equation}

\item
Finally inserting all these conditions into the limit we quickly find
\begin{equation}
G_{\hat r \hat r}|_{\, 0} = G_{\hat \theta \hat \theta}|_{\, 0} = G_{\hat \phi \hat \phi}|_{\, 0}
=  2 \zeta(0)^2 \; R'''(0),
\end{equation}
\begin{equation}
G_{\hat t \hat t}|_{\, 0} =   6 \zeta(0) \zeta''(0)  -6 \zeta(0)^2 \; R'''(0) ,
\end{equation}
and
\begin{equation}
G_{\hat t \hat t}|_{\, 0} + 3 \, G_{\hat r \hat r}|_{\, 0} =   6 \zeta(0) \zeta''(0).
\end{equation}

\item
An analysis based on Taylor series and Laurent series again quickly leads to :
\begin{equation}
R(0) = 0, \qquad R'(0) = 1, \qquad   R''(0) = 0 \qquad \zeta(0) > 0, 
\qquad \mathrm{and} \qquad\zeta'(0) = 0 .
\end{equation}

\end{itemize}
We have now seen the Buchdahl transformation $\zeta(r)\to1/\zeta(r)$ arise in two separate settings --- in both isotropic coordinates and in the Buchdahl coordinates we have just investigated. This raises the question as to whether we can expect to encounter Buchdahl like transformations in more general settings, perhaps by suitably choosing the functional form of the metric. 
This will be our last topic to be investigated in this article.

\section{Generalized Buchdahl ansatz}
We now return to general diagonal coordinates and, based on our insight from dealing with isotropic coordinates and Buchdahl coordinates, make a specific \emph{ansatz} for the functional form of the metric components. Without loss of generality we choose the metric coefficients to be
\begin{equation}
\d s^2 = - \zeta(r)^{2} \; \d t^2 + \zeta(r)^{-2} \left\{ {\d r^2\over E(r) } +  R(r)^2 \d\Omega^2 \right\}.
\end{equation}
We calculate
\begin{equation}
G_{\hat r\hat r} =  - E (\zeta')^2 - {\zeta^2\;[1- E \,(R')^2]\over R^2},
\end{equation}
\begin{equation}
G_{\hat\theta\hat\theta} =  G_{\hat\phi\hat\phi} = +E (\zeta')^2 + {\zeta^2\;[2E R''+ E' R']\over 2 R} ,
\end{equation}
and
\begin{equation}
G_{\hat t \hat t} =  -3E (\zeta')^2+2E \zeta\zeta'' +\zeta E' \zeta' - {2E\zeta( \zeta R''-   2 R' \zeta' )\over R} 
+ {\zeta^2\;[1-E(R')^2]\over R^2} - {\zeta^2 E' R'\over R} .
\end{equation}
Imposing pressure isotropy supplies us with a first-order homogeneous
ODE for $\zeta(r)$:
\begin{equation} 
\label{ode_for_Buchdahl}
\left({\zeta'\over\zeta}\right)^2 =-{[2-2E(R')^2+2E R R''+R E' R']\over 4 E R^2}.
\end{equation}
This is very similar to the equation we obtained in isotropic coordinates and Buchdahl coordinates, but now in general diagonal coordinates --- the key point is that we have carefully chosen the functional form of the metric components. Rearranging this into an ODE for $E(r)$ yields a first-order linear differential equation
\begin{equation}
[\zeta^2 R R' ] E' + [  4 R^2 (\zeta')^2 + 2 R \zeta^2 R'' - 2 \zeta^2 (R')^2 ] E + 2 \zeta^2 = 0.
\end{equation}
(In contrast, rearranging this into an ODE for $R(r)$ yields a second-order nonlinear differential equation which does not seem to be particularly useful.)

\begin{theorem}[\bf Generalized Buchdahl]
  If $\{\zeta(r),E(r),R(r)\}$ describes a perfect fluid then so does
  $\{\zeta(r)^{-1},E(r),R(r)\}$. This is the Buchdahl transformation in yet
  another disguise (now in its most general setting).  The geometry defined by holding $E(r)$ and $R(r)$ fixed and
 transforming the given perfect fluid sphere
 \begin{equation}
\d s^2 = - \zeta(r)^{2} \; \d t^2 + \zeta(r)^{-2} \left\{ {\d r^2\over E(r) } +  R(r)^2 \d\Omega^2 \right\}.
\end{equation}
into the new geometry
\begin{equation}
\d s^2 = - \zeta(r)^{-2} \; \d t^2 + \zeta(r)^{2} \left\{ {\d r^2\over E(r) } +  R(r)^2 \d\Omega^2 \right\}.
\end{equation}
yields a spacetime that is also a perfect fluid sphere. 
That is, the mapping
\begin{equation}
\gth{Generalized Buchdahl}: \left\{ \zeta, E, R  \right\} \mapsto \left\{\zeta^{-1} , E, R \right\}
\end{equation}
takes perfect fluid spheres into perfect fluid spheres.
\end{theorem}

\begin{proof}
By inspection. (Note very strong similarities to the discussion for isotropic
coordinates and Buchdahl coordinates. Note that because we are now working in general diagonal coordinates, where the metric has three independent components, this is the most general result we can hope for for a static spherically symmetric spacetime.)
\end{proof}

The regularity conditions in this case can be extracted from a minor variant of the previous analysis, and reduce to:
\begin{equation}
R(0) = 0, \qquad E(0) [R'(0)]^2 = 1, \qquad  E'(0) = - 2 R''(0)/[R'(0)]^3 = 0, \qquad
\zeta(0) > 0, \qquad \mathrm{and} \qquad \zeta'(0) = 0.
 \end{equation}
Finally inserting all these conditions into the central limit we quickly find
\begin{equation}
G_{\hat r \hat r}|_{\, 0} = G_{\hat \theta \hat \theta}|_{\, 0} = G_{\hat \phi \hat \phi}|_{\, 0}
=  {\zeta(0)^2 \; R'''(0)\over R'(0)^3} - {3\zeta(0)^2 R''(0)^2\over R'(0)^4} + {\zeta(0)^2 E''(0)\over 2},
\end{equation}
\begin{equation}
G_{\hat t \hat t}|_{\, 0} =   
{-3\zeta(0)^2 \; R'''(0)\over R'(0)^3} + {9\zeta(0)^2 R''(0)^2\over R'(0)^4} - {3\zeta(0)^2 E''(0)\over 2}
+{6 \zeta(0) \zeta''(0)\over R'(0)^2},
\end{equation}
and
\begin{equation}
G_{\hat t \hat t}|_{\, 0} + 3 \, G_{\hat r \hat r}|_{\, 0} =   {6 \zeta(0) \zeta''(0)\over R'(0)^2}.
\end{equation}

In addition, we could use the fact that the isotropy condition, when viewed as a differential equation in $E(r)$, is first-order linear to develop yet another variant on the theorem {\bf General diagonal 1}. As no new significant insight is gained we suppress the details.

\section{Discussion}

In this article, we have developed several new and significant transformation theorems that map perfect fluid spheres to perfect fluid spheres using both ``usual'' and ``unusual'' coordinate systems --- such as Schwarzschild (curvature), isotropic, Gaussian polar (proper radius), Synge isothermal (tortoise), and Buchdahl coordinates. In each case we developed at least one such transformation
theorem, while in several cases we have been able to develop multiple transformation
theorems.  

We have also investigated regularity conditions at the centre of the fluid sphere in all these coordinate systems. The meaning of regularity is that (at a minimum) we want the spacetime geometry to be well-defined and smooth everywhere. We can calculate the regularity conditions by hand, but note that  there is another way to more mechanically derive these regularity conditions by using Taylor and Laurent series, with calculations aided by symbolic manipulation programs such as {\sf Maple}. We have found it convenient to use this observation to provide a crosscheck on our explicit computations.

In summary, we have now extended the algorithmic technique originally introduced in reference~\cite{Rahman} and extended in references~\cite{Lake} and~\cite{Martin1, Petarpa1, Petarpa2,  Petarpa:thesis,  Petarpa:Greece, Petarpa:MG11} to many other coordinate systems and many other functional forms for the metric.

\section*{Acknowledgements}

This research was supported by the Marsden Fund administered by the Royal Society of New Zealand. 

PB was supported by a Royal Thai Scholarship and a Victoria University Small Research Grant.



\end{document}